\documentclass[11pt]{article}
\usepackage{array}
\usepackage{amsmath}
\usepackage{amsfonts}
\usepackage{amssymb}
\usepackage{epsfig}

\newtheorem{theorem}{Theorem}
\newtheorem{corollary}{Corollary}

\newtheorem{lemma}{Lemma}
\newtheorem{proposition}{Proposition}
\newtheorem{remark}{Remark}
\newenvironment{proof}[1][Proof]{\textbf{#1.} }{\ \rule{0.5em}{0.5em}}

\begin{document}

\vspace{0.8pc}
\centerline{\large\bf SMOOTHED EXTREME VALUE ESTIMATORS}
\vspace{2pt}
\centerline{\large\bf OF NON-UNIFORM POINT PROCESSES BOUNDARIES}
\vspace{2pt}
\centerline{\large\bf WITH APPLICATION TO STAR-SHAPED SUPPORTS ESTIMATION}
\vspace{2pt}
\vspace{.4cm}
\centerline{St\'ephane Girard$^{(1)}$ and Ludovic Menneteau$^{(2)}$}
\vspace{.4cm}
\centerline{\it $^{(1)}$ INRIA Rh\^one-Alpes, $^{(2)}$ Universit\'e Montpellier 2}
\vspace{.55cm}

\begin{quotation}
\noindent {\it Abstract:}
We address the problem of estimating the edge of a bounded set 
in $\mathbb{R}^{d}$ given a random set of points drawn from the interior.
Our method is based on a transformation of estimators dedicated
to uniform point processes and obtained by smoothing some
of its bias corrected extreme points.  
An application to the estimation of star-shaped
supports is presented.
\par

\vspace{9pt}
\par


\end{quotation}\par

\section{Introduction}

We address the problem of estimating a bounded set $S$ of $\mathbb{R}^{d}$ given
a finite random set $S_n$ of points drawn from the interior. This kind of problem
arises in various frameworks such as
classification (Hardy and Rasson~(1982)),
image processing (Korostelev and Tsybakov~(1993))
 or econometrics problems (Deprins~(1984)).
A lot of different solutions were proposed
since Geffroy~(1964) and Renyi and Sulanke~(1963)
 depending on the properties of
the observed random set $S_n$ and of the unknown set $S$.
Up to our knowledge, the set valued estimators of~Chevalier~(1976),
Gensbittel~(1979) and of Devroye and Wise~(1980)~are the
more general in the sense that they require little assumptions
on $S_n$ and $S$.  
Recently (Girard and Menneteau~(2005), Menneteau~(2007)),
estimators have been introduced for estimating supports writing
$$
S=\{(x,y)\in E\times \mathbb{R},0\leq y\leq f(x)\},
$$
where $f$ is an unknown function and $E$ is a given subset of $\mathbb{R}^{d-1}$.
Thus, the estimation of $S$ reduces to the estimation of the function $f$. 
These methods assume that the random set $S_n$ is obtained
from a point process with mean measure independent from $y$.
In this paper, we propose an extension of the estimators in order
to overcome this limitation. In section~\ref{estima}, the
new family of estimators is introduced. Section~\ref{asymp} is 
devoted to their asymptotic properties. We state a multivariate central
limit theorem as well as a moderate deviations principle.
 These results are applied in section~\ref{secpolar} to the estimation
of star-shaped supports.  Proofs
are collected in section~\ref{secproofs}.

\section{Boundary estimators}
\label{estima}

Let $(E,\mathcal{E},\nu)$ be a probability space, with $E\subset \mathbb{R}^{d-1}$
and where $\nu$ is absolutely continuous with respect to $\lambda$ the Lebesgue measure on $\mathbb{R}^{d-1}$.  
Let $f:(E,\mathcal{E})\rightarrow (\mathbb{R}^+,\mathcal{B}(\mathbb{R}^+))$
be a measurable function, where $\mathcal{B}\left(  \mathbb{R}\right)  $ is the
Borel $\sigma$-algebra on $\mathbb{R}$. Consider the set
\begin{equation}
\label{eqsupport}
S=\{(x,y)\in E\times \mathbb{R},0\leq y\leq f(x)\}.
\end{equation}
Our aim is to estimate $S$ from
a sequence of $S$-valued random vectors
$$
S_n=\{(X_{n,i},Y_{n,i}),\;1\leq i\leq N_{n}(S)\},
$$
with associated counting process 
$$
N_{n}=\left\{  N_{n}\left(  D\right)  :D\in\mathcal{E}\otimes\mathcal{B}%
\left(  \mathbb{R}^+\right)  \right\}  ,\;\;n\geq 1,
$$
of mean measure
\begin{equation}
\label{eqmesure}
n\;c\; \varphi\left(  x,y\right)  \;\mathbf{1}_{S}\left(  x,y\right)
\;\nu(dx)\;dy,
\end{equation}
where $\varphi: S\to\mathbb{R}^+$ is a given non negative function,
and $c$ is an unknown positive parameter.  
In the following, some additional hypothesis are introduced on $\varphi$. 
Two cases are considered below:\\
(P) $N_n$ is a Poisson point process,\\
(E) $N_n$ is an ($n$-sample) empirical point process.\\
In view of~(\ref{eqsupport}), it appears that 
the estimation of the support~$S$ is equivalent to the estimation of the frontier~$f$.
We refer to section~\ref{secpolar} for an illustrative example
of this framework. It is
shown that the estimation of star-shaped supports of homogeneous
point processes reduces to the estimation of 
supports~(\ref{eqsupport}) associated to point processes with
mean measure~(\ref{eqmesure}).

\noindent The estimators proposed in this paper are based on 
a measurable partition of $E$,\\
$\{I_{n,r}:\;1\leq r\leq k_{n}\}$, with
$k_{n}\uparrow\infty$.  
For all $1\leq r\leq k_{n}$, we note 
\[
D_{n,r}=\{(x,y):x\in I_{n,r},\;0\leq y\leq f(x)\}
\]
the cell of $S$ built on $I_{n,r}$ and $N_{n,r}=N_{n}(D_{n,r})$. 
Let us introduce the conditional quantile transformation
\[
\Phi_{x}:y\in\mathbb{R}^{+}\mapsto\int_{0}^{y}\varphi\left(  x,t\right)  dt \in \mathbb{R}^{+}
\]
and the extreme points
\[
(X_{n,r}^{\ast},Y_{n,r}^{\ast})=\underset{(X_{n,i},Y_{n,i})\in D_{n,r}}{\arg\max}
\Phi_{X_{n,i}}\left( Y_{n,i}\right)  ,
\]
if $N_{n,r}\neq0$ and $(X_{n,r}^{\ast},Y_{n,r}^{\ast})=(0,0)$ otherwise.
In the following, the convention $0\times\infty=0$ is adopted. 
\noindent Our estimator of $f(x)$ is:
\begin{equation}
\label{defest}
\hat{f}_{n}(x;\hat c_n)  =\Phi_{x}^{-1}\left(  \sum_{r=1}^{k_{n}}%
\nu_{n,r}\kappa_{n,r}(x)\left(   \Phi_{X_{n,r}^*}(Y_{n,r}^{\ast} )
+\frac{1}{n\hat c_n(x)\nu_{n,r}}\right) \right) ,
\end{equation}
where 
$\nu_{n,r}=\nu(I_{n,r})$,
$\kappa_{n,r}:E\rightarrow\mathbb{R}$ is a weighting function
determining the nature of the smoothing introduced in the estimator,
and $\hat c_n(x)$ is a convenient estimator of $c$.  
Some examples are provided in section~\ref{secpolar}.
\begin{remark}
\label{remone}
When $\varphi=1$, 
$\hat{f}_{n}$ is the estimator defined in Menneteau~(2007):
\begin{equation}
\label{estimennet}
\hat{f}_{n}(x;\hat c_n)  = \sum_{r=1}^{k_{n}}
\nu_{n,r}\kappa_{n,r}(x)\left( Y_{n,r}^{\ast} 
+\frac{1}{n\hat c_n(x)\nu_{n,r}}\right). 
\end{equation}
It can be seen that $Y_{n,r}^{\ast}$ is an estimator of the maximum of $f$ on
$I_{n,r}$ with negative bias. The use of the random variable
$1/(n\hat c_n(x)\nu_{n,r})$ allows to reduce this bias,
see also Girard and Menneteau~(2005) for an example.
\end{remark}
Our estimator~(\ref{defest}) can be considered as a transformation back-transformation
of~(\ref{estimennet}). The first transformation allows to obtain extreme values
$\Phi_{X_{n,r}^*}(Y_{n,r}^{\ast})$ of an homogeneous point process, while the
back-transformation, via $\Phi_x^{-1}$, gives back an estimation of the frontier
of the original non-uniform point process.
The next section is devoted to the asymptotic properties of $\hat{f}_n$.  
General conditions are imposed to the partition $(I_{n,r})$,
the functions $\kappa_{n,r}$, $\hat c_n$ and $\Phi$ to obtain
a central limit theorem and a moderate deviations principle for $\hat f_n$.  

\section {Main results}
\label{asymp}

Let us introduce some auxiliary functions, defined for all $x\in E$:
$$
g(x)=\Phi_{x}(f(x))  =\int_{0}^{f\left(  x\right) }\varphi(x,t)  dt
$$
is the frontier function of the	homogenized point process.
Let $w_{n,r}(x)=\kappa_{n,r}\left(  x\right)  /\kappa_{n}(x)$ be the
renormalized weights where we have defined 
$$
\kappa_{n}(x)=\left(  \sum_{r=1}^{k_{n}}\kappa_{n,r}^{2}(x)\right)^{1/2}.
$$
Define $\nu_{n}=\min\{\nu_{n,r},\;1\leq r\leq k_{n}\}$,
$m=\inf\{g(x), x\in E\}$ and
$M=\sup\{g(x), x\in E\}$.   
Let us also introduce the step function, defined for all $x\in E$ by
$$
g_n(x)= \sum_{r=1}^{k_{n}}\kappa_{n,r}(x)\int_{I_{n,r}} g\;d\nu.
$$
First assumptions are devoted to the function $\varphi$:\\
\noindent $(\mathrm{\Phi })$ $\varphi$ is continuous on $\stackrel{o}{S}$,
 positive almost everywhere on ${S}$, $\varphi(x,f(x))>0$ for all $x\in E$ and
$y\to\varphi(x,y)$ is left-differentiable at $y=f(x)$.
\begin{remark}
\label{remtwo}
\noindent Under assumption $(\mathrm{\Phi })$, $\varphi$ can be extended to
$E\times\mathbb{R}^+$ such that
for all $x\in E$,\\ 
\hspace*{1cm} i) $y\to\varphi(x,y)$ is continuous at $y=f(x)$,\\
\hspace*{1cm} ii) $y\to\frac{\partial{\varphi}(x,y)}{\partial y}$
is continuous at $y=f(x)$.\\
In the sequel, this kind of extensions will be still denoted by $\varphi$.
\end{remark}
Let $(\varepsilon_n)$ be a sequence of positive
real numbers such that $\varepsilon_n=1$ or $\varepsilon_n\downarrow 0$.  
The following assumptions will reveal useful to control
the asymptotic behavior of $\hat{f}_n$.  \\
$\left(  \mathrm{H.1}\right)  \;k_{n}\uparrow\infty$
and $(n\nu_{n})^{-1} \max(\log(n),\varepsilon_n^{-1})\rightarrow 0$ as $n\rightarrow\infty$.\medskip
\newline $\left(  \mathrm{H.2}\right)  \;$~$0<m\leq M<+\infty$ and
\[
\delta_{n}:=\max_{1\leq r\leq k_{n}} \nu_{n,r}
\sup_{(t,s)\in I_{n,r}^2} (g(t)-g(s))
=o\left(
{1/n}\right)  \;\mathrm{as\;}n\rightarrow\infty.
\]
There exists $F\subset E$ such that \medskip\newline $\left(  \mathrm{H.3}%
\right)  \;$For each $(x_{1},...,x_{p})\subset F$, there exists a regular
covariance matrix\\ $\Sigma_{(x_{1},\dots,x_{p})}=[\sigma(x_{i},x_{j})]_{1\leq
i,j\leq p}$ in ${}\mathbb{R}^{p}$ such that for all $1\leq i$, $j\leq p$,
\[
\sum_{r=1}^{k_{n}}w_{n,r}(x_{i})w_{n,r}(x_{j})\rightarrow\;\sigma(x_{i}%
,x_{j})\;\mathrm{as\;}n\rightarrow\infty.
\]
$\left(  \mathrm{H.4}\right)  \;$For all $x\in F$,
\[
\varepsilon_n^{-1/2}\max_{1\leq r\leq k_{n}}|w_{n,r}(x)|\rightarrow0\;\mathrm{as\;}n\rightarrow
\infty.
\]
$\left(  \mathrm{H.5}\right)  \;$For all $x\in F$,
\[
\varepsilon_n^{1/2}\left|  g_n(x)-g(x)\right|
=o\left(  {\frac{\kappa_{n}(x)}{n}}\right)  \;\mathrm{as\;}n\rightarrow\infty.
\]
$\left(  \mathrm{H.6}\right)  \;$For all $x\in F$,
\[
\varepsilon_n^{1/2}\sum_{r=1}^{k_{n}}|w_{n,r}(x)| \left(  n\delta_{n}\right)^{2}
\rightarrow0
\;\mathrm{as\;}n\rightarrow\infty.
\]
$\left(  \mathrm{H.7}\right)\;$ Either $\varphi$ is a constant function, or for all $x\in F$, 
$$
\varepsilon_n^{-1/2}\kappa_n(x)/n\to 0
\;\mathrm{as\;}n\rightarrow\infty.
$$ 

\noindent Before proceeding, let us comment on the assumptions. $\left(
\mathrm{H.1}\right)  $--$\left(  \mathrm{H.4}\right)  $~are devoted to the
control of the centered estimator. Assumption $\left(\mathrm{H.1}\right)$~imposes that
the mean number of points in each
cell goes to infinity. $\left(  \mathrm{H.2}\right)  $\ requires the
unknown function $g$ to be bounded away from 0. It also imposes that the mean
number of points in the cell $D_{n,r}$ above $m_{n,r}$ converges to 0. Note that
$\left(  \mathrm{H.1}\right)  $~and $\left(  \mathrm{H.2}\right)  $~force the
oscillation of $g$ on $I_{n,r}$ to converge uniformly to 0.
$\left(  \mathrm{H.3}
\right)  $ is devoted to the multivariate aspects of the limit theorems.
$\left(  \mathrm{H.4}\right)  $ imposes to the weight functions $\kappa
_{n,r}(x)$ in the linear combination (\ref{defest}) to be approximatively of
the same order. This is a natural condition to obtain an asymptotic Gaussian
behavior. 
Assumptions $\left(  \mathrm{H.5}\right)  $~and $\left(  \mathrm{H.6}\right)
$~are devoted to the control of the bias term $\mathbb{E}(\widehat{f}%
_{n}\left(  x\right)  )-f(x)$. They prevent it to be too important with
respect to the variance of the estimate (which will reveal to be of order
$\kappa_{n}(x)/n$). 
Finally,
$\left(  \mathrm{H.6}\right)  $~can be looked at as a stronger version of
$\left(  \mathrm{H.2}\right)  $.

The last assumptions control the estimation of $c$.\\
$\left(  \mathrm{C.1}\right)  \;$For all $x\in F$, and any $\eta>0$
\[
\mathop{\lim\sup}_{n\to\infty}\varepsilon_n \log P\left(\varepsilon_n^{1/2}\left|
\sum_{r=1}^{k_{n}}w_{n,r}(x)\right|\left|\hat c_n(x)^{-1}-c^{-1}\right|
\geq\eta\right)
=-\infty.
\]
$\left(  \mathrm{C.2}\right)  \;$For all $x\in F$, and any $\eta>0$
\[
\mathop{\lim\sup}_{n\to\infty}\varepsilon_n \log P\left(\left|
\hat c_n(x)-c\right|\geq \eta \right) =-\infty.
\]
Condition (C.1) imposes the speed of convergence of the estimator $\hat c_n$
towards the unknown parameter $c$ in order to cancel the bias term $-1/(nc)$, see
Remark~\ref{remone}. Assumption (C.2) allows to replace $c$ by
its estimator in the asymptotic variance of $\hat f_n$.
\noindent Our first results state the multivariate central limit
theorem for $\hat f_n$.  

\begin{theorem}
\label{TCLvarphi} Let $\varepsilon_n=1$ and suppose $(\Phi)$, $\left(
\mathrm{H.1}\right)  $-$\left(  \mathrm{H.7}\right)$ are verified.   
Let $\hat c_{1,n}$ and $\hat c_{2,n}$ verifying respectively $(\mathrm{C.1})$
and $(\mathrm{C.2})$. 
For all $\left(
x_{1},...,x_{p}\right)  \subset F,$%
\[
\left\{  \frac{n\hat c_{2,n}(x_j)\varphi\left(  x_{j},f\left(  x_{j};\hat c_{1,n}\right)  \right)  }%
{\kappa_{n}\left(  x_{j}\right)  }\left(  \hat{f}_{n}\left(  x_{j};\hat c_{1,n}\right)
-f\left(  x_{j}\right)  \right)  :1\leq j\leq p\right\}  \underset
{\mathcal{D}}{\rightarrow}N\left(  0,\Sigma_{\left(  x_{1},...,x_{p}\right)
}\right)  ,
\]
where  $N\left(
0,\Sigma_{\left(  x_{1},...,x_{p}\right)  }\right)  $ is the centered Gaussian
distribution in $\mathbb{R}^{p},$ with covariance matrix $\Sigma_{\left(
x_{1},...,x_{p}\right)  }.$
\end{theorem}
\begin{corollary}
\label{coromain} Theorem~\ref{TCLvarphi} 
holds when 
$\varphi\left(  x_{j},f\left(  x_{j}\right)  \right)$  is
replaced by $\varphi(  x_{j},\hat{f}_{n}(  x_{j};\hat c_{1,n}))$.
\end{corollary}
This leads to an explicit asymptotic $\gamma\%$ confidence interval
for $f(x)$:
\[
\left[
\hat{f}_n(x;\hat c_{1,n})-z_{\gamma}
\frac{\kappa_n(x)}{n\hat c_{2,n}(x)\varphi(x,\hat{f}_n(x;\hat c_{1,n}))},
\hat{f}_n(x;\hat c_{1,n})+z_{\gamma}
\frac{\kappa_n(x)}{n\hat c_{2,n}(x)\varphi(x,\hat{f}_n(x;\hat c_{1,n}))}\right]  ,
\]
where $z_{\gamma}$ is the $(\gamma+1)/2$th quantile of the $N(0,1)$
distribution. Note that the computation of this interval does not require a
bootstrap procedure as for instance in Hall {\it et al}~(1998).

\noindent The following family of large deviations principle is sometimes referenced
as a moderate deviations principle (see e.g. Dembo and Zeitouni~(1993)).
\begin{theorem}
\label{devvarphi} 
Let $\varepsilon_n\downarrow 0$ and 
suppose $(\Phi)$, $\left( \mathrm{H.1}\right)  $-$\left(  \mathrm{H.7}\right)$ are verified.  
Let $\hat c_{1,n}$ and $\hat c_{2,n}$ verifying respectively $(\mathrm{C.1})$
and $(\mathrm{C.2})$. 
For all $\left(
x_{1},...,x_{p}\right)  \subset F$
such that $\Sigma_{\left(  x_{1},...,x_{p}\right)}$ is regular,
the sequence of random vectors 
\[
\left\{  \frac{\varepsilon_n^{1/2}n\hat c_{2,n}(x_j)\varphi\left(  x_{j},f\left(  x_{j};\hat c_{1,n}\right)  \right)  }%
{\kappa_{n}\left(  x_{j}\right)  }\left(  \hat{f}_{n}\left(  x_{j};\hat c_{1,n}\right)
-f\left(  x_{j}\right)  \right)  :1\leq j\leq p\right\} 
\]
follows the large deviations principle in $\mathbb{R}^p$ with speed $(\varepsilon_n)$
and good rate function
$$
I_{(x_1,\dots,x_p)}: u\in\mathbb{R}^p\mapsto \frac{1}{2} u\Sigma_{\left(  x_{1},...,x_{p}\right)}^{-1} \;^t u.
$$
\end{theorem}
\begin{corollary}
\label{coromain2} Theorem~\ref{devvarphi} holds when 
$\varphi\left(  x_{j},f\left(  x_{j}\right)  \right)$  is
replaced by $\varphi(  x_{j},\hat{f}_{n}(  x_{j};\hat c_{1,n}))$.
\end{corollary}
As a consequence, one can obtain a rate of convergence in the 
almost sure consistency of the frontier estimator. More precisely,
Corollary~\ref{coromain2} and the Borel-Cantelli Lemma entail that,
for all $x\in E$,
$$
\limsup_{n\to\infty}
  \frac{n\hat c_{2,n}(x_j)\varphi\left(  x_{j},f\left(  x_{j};\hat c_{1,n}\right)  \right)  }
{(2\log n)^{1/2}\kappa_{n}\left(  x_{j}\right)  }\left|  \hat{f}_{n}\left(  x_{j};\hat c_{1,n}\right)
-f\left(  x_{j}\right)  \right| \leq 1 \mbox{ a.s.} 
$$
In terms of confidence interval, Corollary~\ref{coromain2} can also
be useful to compute the logarithmic asymptotic level of confidence
intervals with asymptotic level 0. See Menneteau (2007) for further
details. Finally, in estimation theory, Corollary~\ref{coromain2}
is of interest to compute the Kallenberg efficiency of $\hat f_n$
(Kallenberg, 1983a, 1983b).

\section {Star-shaped supports}
\label{secpolar}

One motivating application of the general framework introduced in
section~\ref{estima} is the estimation of star-shaped supports
in ${\mathbb R}^d$, $d\geq 2$. We refer to Baillo and Cuevas~(2001)
for an adaptation of the estimator defined by Devroye and Wise~(1980)
to this situation.
The support can be parameterized in polar coordinates such as:
\[
S^{\mathrm{pol}}=\left\{  \left(  u,v\right)  = P_d(x,y)
:x\in E,\;
 0\leq y\leq f\left(   x\right)  \right\}  ,
\]
where $E=[0,\pi)^{d-2}\times[0,2\pi)$, $ 
f:E \rightarrow\mathbb{R}^{\mathbb{+}}
$
is a measurable function, and the mapping $P_d:E\times(0,+\infty)\to \mathbb{R}^{d-1}\times (0,+\infty)$
with
$$ P_d(x,y)=y 
\displaystyle\left(\cos x_1,\cos x_2\sin x_1,\dots,\cos x_{d-1}\prod_{j=1}^{d-2}\sin x_j,
\prod_{j=1}^{d-1}\sin x_j\right)^t
$$
defines the polar coordinates (see Mardia {\it et al}~(1979), section 2.4)
in $\mathbb{R}^d$.  We consider the sequence of Poisson or empirical point processes
\[
N_{n}^{\mathrm{pol}}=\left\{  N_{n}^{\mathrm{pol}}\left(  B\right)
:B\in\mathcal{B}\left( S^{\mathrm{pol}}\right)  \right\}  ,\;\;n\geq1,
\]
with mean measure 
$$
n\;c\;\mathbf{1}_{S^{\mathrm{pol}}}(u,v)\;du\;dv,
$$
where $c>0$.
Let $\left(  U_{n,i},V_{n,i}\right)  _{i\geq1}$ be the point process
associated to $N_{n}^{\mathrm{pol}}$. 
Our aim is to estimate $S^{\mathrm{pol}}$ via an estimation of 
the associated frontier function~$f$.  
This function can also be seen as the frontier of the support
\[
S=\left\{  \left(  x,y\right)  :x\in E,\;0\leq y\leq
f\left(  x\right)  \right\} 
\]
of the point process 
$\left(  X_{n,i},Y_{n,i}\right)  _{i\geq1}$ defined by for all $i\geq1,$
\[
\left(  U_{n,i},V_{n,i}\right)  =P_d(X_{n,i},Y_{n,i}),
\]
where $(X_{n,i})$ represents the sequence of polar angles
and $(Y_{n,i})$ the sequence of polar radius.

In the case $d=2$, classical planar polar coordinates are obtained,
see figure~\ref{figdessin} for an illustration. For $d=3$, we get
usual spherical coordinates. Note that, in this situation, cylindrical
coordinates can also enter the framework of section~\ref{asymp}.


It will appear in Lemma~\ref{lemvarphi} in section~\ref{secproofs}, that the point process
$\left(  X_{n,i},Y_{n,i}\right)  _{i\geq1}$ is no more homogeneous
but benefits of the mean measure~(\ref{eqmesure}) with 
$$
\varphi(x,y)=\gamma_d\; y^{d-1} \;\;\mbox{ and }\;\; \nu(dx)=h_d(x)dx,
$$
$$
\mbox{where }\;\;\gamma_d=\int_E\prod_{j=1}^{d-1}(\sin x_j)^{d-1-j}dx\;\; \mbox{ and }
\;\;h_d(x)=\gamma_d^{-1}\prod_{j=1}^{d-1}(\sin x_j)^{d-1-j},
$$
i.e.
\begin{equation}
\label{eqmesurepol}
n\;c\; \gamma_d\; y^{d-1}\;  \mathbf{1}_{S}\left(  x,y\right) \;h_d(x)\;dx\; dy.
\end{equation}
As for choosing the partition, a natural choice would be to consider
equiprobable sets $(I_{n,r})$ with respect to the polar angle distribution. 
Unfortunately, from~(\ref{eqmesurepol}), it is easily seen that the polar angle density is
\begin{equation}
\label{densX}
{h_d(x)f^d(x)}\left/\int_E h_d(t)f^d(t)dt\right.,
\end{equation}
and thus depends on the unknown frontier function $f$. Without prior knowledge
on $f$, one may consider in~(\ref{densX}) that $f$ is a constant. In this case,
the measure induced by~(\ref{densX}) is $\nu$. Moreover, since $f$ is both
bounded from zero and upper bounded,
(\ref{densX}) implies that the polar angle distribution is equivalent to $\nu$.
These considerations lead us to choose a measurable partition of $E$
such that $\nu(I_{n,r})=1/k_n$ for $1\leq r\leq k_{n}$.
In accordance with the notations of section~\ref{estima}, let 
for all $1\leq r\leq k_{n}$, 
\begin{eqnarray*}
D_{n,r}&=&\left\{  \left(  x,y\right)  :x\in I_{n,r},\;0\leq y\leq
f\left(  x\right)  \right\},\\
Y_{n,r}^{\ast}&=&\max\left\{  Y_{n,i}:\left(  X_{n,i},Y_{n,i}\right)  \in
D_{n,r}\right\},\\
\mbox{ and }N_{n,r}&=&N_{n}\left(  D_{n,r}\right).
\end{eqnarray*}

\subsection{A general kernel estimator}
\label{secpolar1}

In the sequel, we adopt the following weight function
\begin{equation}
\label{choixnoyau1}
\kappa_{n,r}(x)= k_n\int_{I_{n,r}}K_{n}(x,t)\nu(dt),
\end{equation}
where $K_n$ is a general smoothing kernel,
and the global estimator of $c$ defined by
\begin{equation}
\label{choixcn}
\hat c_n^{glo}  = \frac{k_n^2}{n} \left(\sum_{r=1}^{k_{n}}
 \frac{\Phi_{X_{n,r}^*}(Y_{n,r}^{\ast} )}{N_{n,r}}\right)^{-1} =
 \frac{d}{\gamma_d} \frac{k_n^2}{n} \left(\sum_{r=1}^{k_{n}}
 \frac{(Y_{n,r}^{\ast} )^d}{N_{n,r}}\right)^{-1} 
\end{equation}
both introduced in Menneteau~(2007).
The framework of section~\ref{estima}
leads to the estimator of the frontier $f$ below (see Lemma~\ref{lemme4} in section~\ref{secproofs}),
\begin{equation}
\label{estipol}
\hat{f}_{n}^{\mathrm{pol}}(x)=\left(
\sum_{r=1}^{k_n}\left(\int_{I_{n,r}}K_{n}(x,t)h_d(t)dt
+ \frac{\int_E K_n(x,t)h_d(t)dt}{k_n N_{n,r}}\right)
\left( Y_{n,r}^{\ast}\right)  ^{d}\right)  ^{1/d},
\end{equation}
and the associated estimator of the support is given by
\[
\hat S_n^{\mathrm{pol}}=\left\{  \left(  u,v\right)  = P_d(x,y)
:x\in E,\;
 0\leq y\leq \hat f_n^{\mathrm{pol}}\left(   x\right)  \right\}.
\]
In this context, Theorem~\ref{TCLvarphi} and Theorem~\ref{devvarphi} permit
to derive the asymptotic behavior of the estimation error 
in the direction $x$ defined as
$\Delta_n(x)= \hat{f}_{n}^{\mathrm{pol}}(  x) -f(x)$.
Let us emphasize that $|\Delta_n(x)|$ can also be interpreted as the
length of the slice in the direction $x$ of the symmetrical difference
between the estimated support $\hat S_n^{\mathrm{pol}}$ and the true
one $S^{\mathrm{pol}}$.
Establishing similar results for the surface of the symmetrical
difference, i.e. the Hausdorff distance, would require uniform 
convergence results, and is thus beyond the scope of this paper.

The following notations will reveal useful to state
the assumptions on $K_n$.
For all $x\in E$ and $1\leq r\leq k_n$, consider the oscillation of
$K_n(x,.)$ over $I_{n,r}$,
\[
\Gamma_{n,r}\left(  x\right)  =\;\sup\left\{  K_{n}(x,t)-K_{n}(x,s):\left(
s,t\right)  \in I_{n,r}\times I_{n,r}\right\},
\]
the smoothing error 
\[
\Psi_{n}\left(  x\right)  =\left\vert \int_{E}K_{n}(x,t)f^d\left(  t\right)
\nu\left(  dt\right)  -f^d\left(  x\right)  \right\vert 
\]
and 
\[
\Xi_{n}\left(  x\right)  =k_{n}\left\vert \sum_{r=1}^{k_{n}}\int
_{I_{n,r}\times I_{n,r}}K_{n}\left(  x,t\right)  \left(  f^d\left(  s\right)
-f^d(t)\right)  \nu\left(  dt\right)  \nu\left(  ds\right)  \right\vert,
\]
which can be interpreted as the loss of information due to the partitioning.
Let us also introduce the maximum oscillation of $f^d$ over each set of the partition
\[
\omega_{n}=\ \underset{1\leq r\leq k_{n}}{\max}\sup_{(s,t)\in I_{n,r}^2}
(f^d(s)-f^d(t)),
\]
and the classical norms
\[
\left\Vert K_n(x,.)\right\Vert _{p}=\left(  \int_{E}\left\vert K_n\left(x,  t\right)
\right\vert ^{p}\nu\left(  dt\right)  \right)  ^{1/p}\;\;\mathrm{and}
\;\;\left\Vert K_n(x,.)\right\Vert _{E}=\,\underset{t\in E}{\sup}\left\vert K_n\left(x,t \right)  \right\vert. 
\]
In this context, the general assumptions (H.3)-(H.7) can be expressed as:
\\
$\left(  \mathrm{K.1}\right)  $ For all
$n\geq1,$ $\int_{E\times E}\left\vert K_{n}(x,t)\right\vert \nu\left(
dx\right)  \nu\left(  dt\right)  <\infty$.
\medskip\newline$\left(
\mathrm{K.2}\right)  \;$For all $\left(  x_{1},x_{2}\right)  \in E\times E$,
\[
\sum_{r=1}^{k_{n}}\Gamma_{n,r}\left(  x_{1}\right)  \int_{I_{n,r}}\left\vert
K_{n}(x_{2},t)\right\vert \nu\left(  dt\right)  =o\left(  \left\Vert
K_{n}(x_{1},\;\mathbf{.}\;)\right\Vert _{2}\left\Vert K_{n}(x_{2}%
,\;\mathbf{.}\;)\right\Vert _{2}\right)  \;\mathrm{as\;}n\rightarrow\infty.
\]
$\left(  \mathrm{K.3}\right)  \;$For all $\left(  x_{1},x_{2}\right)  \in
E\times E$,
\[
\left\langle K_{n}(x_{1},\;\mathbf{.}\;),K_{n}(x_{2},\;\mathbf{.}%
\;)\right\rangle _{2}\left(  \left\Vert K_{n}(x_{1},\;\mathbf{.}\;)\right\Vert
_{2}\left\Vert K_{n}(x_{2},\;\mathbf{.}\;)\right\Vert _{2}\right)
^{-1}\rightarrow\;\sigma(x_{1},x_{2})\;\mathrm{as\;}n\rightarrow\infty.
\]
$\left(  \mathrm{K.4}\right)  \;$For all $x\in E$,
\[
\left(  \varepsilon_{n}k_{n}\right)  ^{-1/2}\left\Vert K_{n}(x,\;\mathbf{.}%
\;)\right\Vert _{2}^{-1}\;\left\Vert K_{n}(x,\;\mathbf{.}\;)\right\Vert
_{E}\rightarrow0\;\mathrm{as\;}n\rightarrow\infty.
\]
$\left(  \mathrm{K.5}\right)  \;$For all $x\in E$,
\[
\varepsilon_{n}^{1/2}nk_{n}^{-1/2}\left\Vert K_{n}(x,\;\mathbf{.}%
\;)\right\Vert _{2}^{-1}\;\max\left(  \Psi_{n}\left(  x\right)  ;\;\Xi
_{n}\left(  x\right)  \right)  \rightarrow0\;\mathrm{as\;}n\rightarrow\infty.
\]
$\left(  \mathrm{K.6}\right)  \;$For all $x\in E$,
\[
\varepsilon_{n}^{1/4}nk_{n}^{-3/4}\left\Vert K_{n}(x,\;\mathbf{.}%
\;)\right\Vert _{2}^{-1/2}\left\Vert K_{n}(x,\;\mathbf{.}\;)\right\Vert
_{1}^{1/2}\omega_{n}\rightarrow0\;\mathrm{as\;}n\rightarrow\infty.
\]
$\left(  \mathrm{K.7}\right)  \;$For all $x\in E$,
\[
\varepsilon_{n}^{-1/2}n^{-1}k_{n}^{1/2} \| K_n(x,.) \|_2\to 0
\;\mathrm{as\;}n\rightarrow\infty.
\]
The results established in section~\ref{asymp} yield:
\begin{theorem}
\label{TCLpol} 
 Let $\varepsilon_n=1$ and suppose that $\mathrm{(H.1)}$, $\mathrm{(H.2)}$, $\left(
\mathrm{K.1}\right)  $-$\left(  \mathrm{K.7}\right)$ are verified.  \\ 
For all $\left(
x_{1},...,x_{p}\right)  \subset E,$%
\[
\left\{  nk_n^{-1/2}\|K_n(x_j,.)\|_2^{-1}\; \hat c_{n}^{glo}\;\gamma_d\; f^{d-1}\left(  x_{j}\right)  
\Delta_n\left(  x_{j}\right) 
  :1\leq j\leq p\right\}  \underset
{\mathcal{D}}{\rightarrow}N\left(  0,\Sigma_{\left(  x_{1},...,x_{p}\right)
}\right). 
\]
\end{theorem}
\begin{theorem}
\label{MDpol} 
Let $\varepsilon_n\downarrow 0$ and 
suppose $\mathrm{(H.1)}$, $\mathrm{(H.2)}$, $\left( \mathrm{K.1}\right)  $-$\left(  \mathrm{K.7}\right)$ are verified.  
For all $\left(
x_{1},...,x_{p}\right)  \subset E$
such that $\Sigma_{\left(  x_{1},...,x_{p}\right)}$ is regular,
the sequence of random vectors 
\[
\left\{  \varepsilon_n^{1/2}nk_n^{-1/2}\|K_n(x_j,.)\|_2^{-1} \;\hat c_{n}^{glo}\;\gamma_d \; f^{d-1}\left(  x_{j}\right)  
\Delta_n(x_j) 
 :1\leq j\leq p\right\} 
\]
follows the large deviations principle in $\mathbb{R}^p$ with speed $(\varepsilon_n)$
and good rate function $I$.
\end{theorem}

\subsection{Illustration in the bi-dimensional case}
\label{secpolar2}

\noindent As an illustration, we consider the case $d=2$.
In this situation, 
$h_2(x)=(2\pi)^{-1}$. 
For the sake of simplicity, we focus on the case where the
partition is equidistant i.e. $I_{n,r}=[2\pi(r-1)k_n^{-1},\; 2\pi r k_n^{-1})$,
$r=1,\dots,k_n$.
For periodicity reasons, we consider the
Dirichlet's kernel
\begin{equation}
K_{n}^D\left(  x,t\right)  =\sum_{j=0}^{\ell_{n}}e_{j}(x)e_{j}(t),\qquad
(x,t)\in [0,2\pi]^{2},\label{Dkern}%
\end{equation}
associated to the trigonometric basis (Tolstov~(1976)):
\begin{equation}
\label{choixnoyau3}
e_{0}(x)=(2\pi)^{-1},\qquad e_{2j-1}(x)={\pi}^{-1}\cos{(j x)},\qquad e_{2j}%
(x)={\pi}^{-1}\sin{(jx)},\quad j\geq1.
\end{equation}
It is well-known that the Dirichlet's kernel can be rewritten as
\begin{eqnarray*}
K_n^D(x,t) &=& \frac{\sin{(2^{-1}(1+\ell_n)(x-t))}}{\sin{(2^{-1}(x-t))}} \mbox{ if } x\neq t\\
&=& 1+\ell_n \mbox{ if } x=t.
\end{eqnarray*}
Since, for all $x\in E$, $\int_E K_n^D(x,t)h_2(t)dt=1$,
the estimator~(\ref{estipol}) becomes
\begin{eqnarray*}
\hat{f}_{n}^{\mathrm{pol}}(x)&=&\left(
\sum_{r=1}^{k_n}\left(\frac{1}{2\pi}\int_{I_{n,r}}K_n^D(x,t)dt
+ \frac{1}{k_n N_{n,r}}\right)
\left( Y_{n,r}^{\ast}\right)  ^{2}\right)  ^{1/2},\\
&=&\left(
\frac{1}{k_n}\sum_{r=1}^{k_n}\left(\int_{r-1}^r K_n^D(x,2\pi k_n^{-1} s)ds
+ \frac{1}{N_{n,r}}\right)
\left( Y_{n,r}^{\ast}\right)  ^{2}\right)  ^{1/2}.
\end{eqnarray*}
In the above context, we have the following result.
\begin{corollary}
\label{coropol}
Suppose $f$ is $C^2$ with
$f(0)=f(2\pi)$ and $f'(0)=f'(2\pi)$.  
Assume that 
(i)~$n^{-1}k_{n}\log(n) =o(1)$,
(ii)~$\ell_{n}\log(\ell_{n}) k_{n}^{-1}=o(1)$,
(iii)~$nk_{n}^{-1/2}\ell_{n}^{-2}=O(1)$,\\
(iv)~$n k_{n}^{-5/2}\ell_n^{1/2}\log(\ell_n)=o(1)$ and
(v)~$n k_{n}^{-7/4}\ell_n^{-1/4}(\log(\ell_n))^{1/2}=o(1)$.\\
Then, for all $\left(  x_{1},...,x_{p}%
\right)  \subset\left[  0,2\pi\right)  ,$%
\begin{equation}
\left\{  v_n \hat c_n^{glo} \hat{f}^{\mathrm{pol}}_n(x_j)\;\Delta_{n}
(x_{j}):1\leq j\leq p\right\}  \underset{\mathcal{D}}{\rightarrow
}N\left(  0,I_{p}\right)  , \label{TriKern}%
\end{equation}
where $v_n=n(\ell_{n}k_{n})^{-1/2}$.  
The choice $\ell_{n}=n^{{10}/{27}}$ and
$k_{n}=n^{{14}/{27}}(\log\left( n\right))^{2/7}  u_{n}^{2}$ 
leads to $v_n= n^{{5}/{9}} \log\left(  n\right)^{-1/7}u_{n}^{-1}$,
where $u_{n} \rightarrow\infty$ arbitrarily slowly.  
\end{corollary}
Since our estimator is based on extreme values, it reaches an asymptotic
convergence rate larger than the classical parametric rate $n^{1/2}$.
At the opposite, estimators built on nonparametric regression techniques
would be limited to convergence rates lower than $n^{1/2}$.
As an example, the optimal convergence rate for estimating $C^2$ regression
functions is $n^{2/5}$~(Stone (1982)).

\subsection{Numerical experiments}

\noindent To conclude, we propose a simple illustration of the behavior
of the estimator $\hat{f}^{\mathrm{pol}}_n$ on a finite sample
situation.   The true frontier function is the $\pi/3$- periodic function
$$
f(x)=1+\exp(-\cos(3 x)),\;x\in[0,2\pi).
$$

\noindent The experiment involves several steps:

-- First, $m=100$ replications of a Poisson process (situation (P))
are simulated with $c=1/\int_0^{2\pi}f(x)dx$ and $n=100$.

-- For each of the $m$ previous set of points, the trigonometric
estimator $\hat{f}_{n}^{\mathrm{pol}}$ is computed with
$h_n=15$ and $k_n=20$.

-- The $m$ associated $L_1$ distances to $f$ are evaluated on a grid.  

-- Finally, the best situation
({\it i.e.} the estimation corresponding to the smallest $L_1$ error)
is represented on Figure~\ref{figsim} and the worst situation
({\it i.e.} the estimation corresponding to the largest $L_1$ error)
is represented on Figure~\ref{figsim2}.

\noindent The results are visually satisfying. More precisely, denoting
by $\xi_n$ the relative $L_1-$ error defined by
$$
\xi_n=\frac{\int_0^{2\pi} |\hat{f}_{n}^{\mathrm{pol}}(x)-f(x)|dx}{\int_0^{2\pi} f(x)dx},
$$
the maximum observed value of $\xi_n$ is $9.6\%$ (corresponding to
Figure~\ref{figsim2}), the minimum observed value is $3.8\%$
(corresponding to Figure~\ref{figsim}) and the mean value is $6.2\%$.

\section {Proofs}
\label{secproofs}

\subsection{Proofs of section~\ref{asymp}}

The proofs of Theorem~\ref{TCLvarphi} and Theorem~\ref{devvarphi}
follow the same lines. They are based on results of Menneteau~(2007)
for homogeneous processes and on an approximation argument. 

1. First, we show in Lemma~\ref{lemKing} that one can associate to
$N_n$ a homogeneous process thanks to a convenient transformation.  
More precisely, let $(\Pi_n)_{n\geq 1}$ denote the sequence of
counting processes defined by
\[
\Pi_{n}:D\in\mathcal{E}\otimes\mathcal{B}\left(  \mathbb{R}\right)
\mapsto\#\left\{  \left(  X_{n,i},\Phi_{X_{n,i}}\left(
Y_{n,i}\right) \right)  \in D\right\}  .
\]
\begin{lemma}
\label{lemKing}
Suppose {\rm ($\Phi$)} holds. Then, 
in situation (P) (resp. (E)),  $\Pi_{n}$ is associated
with a Poisson (resp. an empirical) process on $E\times\mathbb{R}^{+},$ with mean measure
$
nc\;\mathbf{1}_{G}\left(  x,v\right)  \;\nu(dx)\;dv\,,
$
where
\begin{equation}
G=\{(x,v):x\in E\;;\;0\leq v\leq~g(x)\}.
\end{equation}
\end{lemma}
\begin{proof}
In situation (P), the result follows from the Mapping Theorem (see 
Kingman~(1993), p.~18).
In situation (E), the result is obtained by a simple change of variable
(see Cohn~(1980), Theorem~6.1.6).
\end{proof}

2.  As previously remarked in section~\ref{estima}, asymptotic results were
already established for homogeneous processes.  
For convenience of notation, we write $\hat c_{1,n}(x)=\hat c_n$ and
$\hat f_n(x)=\hat f_n(x;\hat c_n)$.  
Following~(\ref{estimennet}), we define for $x\in E$ : 
$$
\hat g_n(x)=\hat{g}_{n}\left( x;\hat c_n\right)     =\Phi_{x}(\hat{f}_{n}(x))  
=\sum_{r=1}^{k_{n}}%
\nu_{n,r}\kappa_{n,r}(x)\left(   \Phi_{X_{n,r}^*}(Y_{n,r}^{\ast} )
+\frac{1}{n\hat c_n(x)\nu_{n,r}}\right) 
$$
an estimator of $g(x)$, the frontier of the homogeneous process.     
Therefore, one can apply to $\hat g_n$ the following results, proved in 
Menneteau~(2007), which assert that
Theorem~\ref{TCLvarphi} and
Theorem~\ref{devvarphi} hold with $\varphi=1$.  
\begin{proposition} 
\label{propdebut}
i) Let $\varepsilon_n=1$ and suppose $(\Phi)$, $\left(
\mathrm{H.1}\right)  $-$\left(  \mathrm{H.6}\right)$ are verified.   
Let $\hat c_{1,n}$ and $\hat c_{2,n}$ verifying respectively $(\mathrm{C.1})$
and $(\mathrm{C.2})$. Then,
for all $\left(
x_{1},...,x_{p}\right)  \subset E,$%
\[
\left\{  \frac{n\hat c_{2,n}(x_j)}
{\kappa_{n}\left(  x_{j}\right)  }\left(  \hat{g}_{n}\left(  x_{j};\hat c_{1,n}\right)
-g\left(  x_{j}\right)  \right)  :1\leq j\leq p\right\}  \underset
{\mathcal{D}}{\rightarrow}N\left(  0,\Sigma_{\left(  x_{1},...,x_{p}\right)
}\right). 
\]
ii)
Let $\varepsilon_n\downarrow 0$ and 
suppose $(\Phi)$, $\left( \mathrm{H.1}\right)  $-$\left(  \mathrm{H.6}\right)$ are verified.  
Let $\hat c_{1,n}$ and $\hat c_{2,n}$ verifying respectively $(\mathrm{C.1})$
and $(\mathrm{C.2})$. 
For all $\left(
x_{1},...,x_{p}\right)  \subset E$
such that $\Sigma_{\left(  x_{1},...,x_{p}\right)}$ is regular,
the sequence of random vectors 
\[
\left\{  \frac{\varepsilon_n^{1/2}n\hat c_{2,n}(x_j)}%
{\kappa_{n}\left(  x_{j}\right)  }\left(  \hat{g}_{n}\left(  x_{j};\hat c_{1,n}\right)
-g\left(  x_{j}\right)  \right)  :1\leq j\leq p\right\} 
\]
follows the large deviations principle in $\mathbb{R}^p$ with speed $(\varepsilon_n)$
and good rate function $I$.

\end{proposition}

3.
We now derive the asymptotic behavior of $\hat f_n$ from that of $\hat g_n$
by an approximation argument given in the next lemma.
\begin{lemma}
\label{lemmapp}
If $\left( \mathrm{H.1}\right)  $-$\left(  \mathrm{H.7}\right)$ hold  
and $\hat c_n$ verifies $(\mathrm{C.1})$, then
\begin{equation}
\label{majder2}
\mathop{\lim\sup}_{n\to\infty}\varepsilon_n\log P\left(
\varepsilon_n^{1/2} \frac{nc}{\kappa_n(x)}
\left|\varphi(x,f(x))(\hat{f}_n(x;\hat c_n)-f(x))-(\hat{g}_n(x;\hat c_n)-g(x))\right|\geq\eta\right)=-\infty.  
\end{equation}
\end{lemma}
\begin{proof}
The result is straightforward if $\varphi$ is a constant function.
We thus focus on the case where, by (H.7),  
\begin{equation}
\label{eqtempo}
\varepsilon_n^{-1/2}\kappa_n(x)/n\to 0
\end{equation}
as $n\to\infty$.
For all $x\in E$, there exists $h_n(x)\in [\min(g(x),\hat{g}_n(x)),
\max(g(x),\hat{g}_n(x))]$, such that
\begin{eqnarray*}
\hat{f}_n(x)-f(x) &=& \Phi_x^{-1}(g(x)+(\hat{g}_n(x)-g(x))) - \Phi_x^{-1}(g(x))\\
&=& (\Phi_x^{-1})'(g(x)) (\hat{g}_n(x)-g(x)) +
\frac{1}{2}(\Phi_x^{-1})''(h_n(x)) (\hat{g}_n(x)-g(x)) ^2.
\end{eqnarray*}
Remarking that $ (\Phi_x^{-1})'(g(x))=1/\varphi(x,f(x))$, we obtain that
\begin{equation}
\label{eqdev2}
\varphi(x,f(x))(\hat{f}_n(x)-f(x) )=  (\hat{g}_n(x)-g(x)) 
+ \frac{1}{2}\varphi(x,f(x)) (\Phi_x^{-1})''(h_n(x)) (\hat{g}_n(x)-g(x)) ^2.
\end{equation}
Set $\eta>0$ and for all $\alpha>0$ introduce 
\begin{eqnarray*}
I_{n}(x,\alpha,\varphi)&=&\left]g(x)-\alpha \frac{\kappa_n(x)}{nc\varepsilon_n^{1/2}},
g(x)+\alpha \frac{\kappa_n(x)}{nc\varepsilon_n^{1/2}}\right[,\\
M_n(x,\alpha,\varphi)&=&
\sup_{u\in I_n(x,\alpha,\varphi) }
\left|(\Phi_x^{-1})''(u)\right|\\
&=&
\sup_{u\in I_n(x,\alpha,\varphi) }
\frac{\left|\partial \varphi(x,\Phi_x^{-1}(u)) /
\partial y\right|}{\varphi^3(x,\Phi_x^{-1}(u))}.
\end{eqnarray*}
From (\ref{eqdev2}), and since $|h_n(x)-g(x)|\leq |\hat g_n(x)-g(x)|$,
it follows that
\begin{eqnarray*}
&& \varepsilon_n^{1/2}\frac{nc}{\kappa_n(x)}
\left|\varphi(x,f(x))(\hat{f}_n(x)-f(x))-(\hat{g}_n(x)-g(x))\right|
1_{\{\frac{nc\varepsilon_n^{1/2}}{\kappa_n(x)}|\hat g_n(x)-g(x)|<\alpha\}}\\
&\leq &\frac{\kappa_n(x)}{2nc\varepsilon_n^{1/2}} \varphi(x,f(x)) M_n(x,\alpha,\varphi)
\left(\frac{nc\varepsilon_n^{1/2}}{\kappa_n(x)}|\hat g_n(x)-g(x)|\right)^2
1_{\{\frac{nc\varepsilon_n^{1/2}}{\kappa_n(x)}|\hat g_n(x)-g(x)|<\alpha\}}\\
&\leq &\frac{\kappa_n(x)}{2nc\varepsilon_n^{1/2}} \varphi(x,f(x)) M_n(x,\alpha,\varphi) \alpha^2\\
&<&\eta,
\end{eqnarray*}
eventually, since from (\ref{eqtempo}),
$I_n(x,\alpha,\varphi)\to \{g(x)\}$ as $n\to\infty$ and thus
$$
M_n(x,\alpha,\varphi) \to \frac{\left|\partial \varphi(x,f(x)) /
\partial y\right|}{\varphi^3(x,f(x))}.
$$
Consequently, for all large $\alpha$
\begin{eqnarray*}
&&\mathop{\lim\sup}_{n\to\infty}\varepsilon_n\log P\left(
\varepsilon_n^{1/2} \frac{nc}{\kappa_n(x)}
\left|\varphi(x,f(x))(\hat{f}_n(x)-f(x))-(\hat{g}_n(x)-g(x))\right|\geq\eta\right)
\\
&\leq& \mathop{\lim\sup}_{n\to\infty}\varepsilon_n\log P\left(
\frac{nc\varepsilon_n^{1/2}}{\kappa_n(x)}|\hat g_n(x)-g(x)|\geq\alpha\right)\\
&\leq& -\frac{\alpha^2}{2\sigma^2(x)}
\end{eqnarray*}
where $\sigma^2(x)=\sigma(x,x) $ with Proposition~\ref{propdebut}.
Letting $\alpha\to\infty$ gives the result.  
\end{proof}

4. The proofs of the announced results are now straightforward:

\noindent\textbf{Proofs of Theorem~\ref{TCLvarphi} and Theorem~\ref{devvarphi}:} 
a) First, we prove the theorems for $\hat c_{2,n}(x)=c$.  
By Lemma~\ref{lemKing}, we can apply Proposition~\ref{propdebut}
to obtain the expected weak convergence
and moderate deviations principle for $(\hat g_n(x)-g(x))$. 
 From Lemma \ref{lemmapp}, $(\hat g_n(x)-g(x))$ and $\varphi(x,f(x))(\hat f_n(x)-f(x))$
share the same asymptotic behavior in terms of weak convergence
and moderate deviations principles.  \\
b) In the general case, it is sufficient to prove that
\begin{equation}
\label{butun}
\mathop{\lim\sup}_{n\to\infty}\varepsilon_n\log P\left(
\varepsilon_n^{1/2} \frac{n}{\kappa_n(x)}\left|\hat c_{n,2}(x)-c \right|
\left|\hat{f}_n(x)-f(x)\right|\geq\eta\right)=-\infty.  
\end{equation}
To this aim, observe that, for all large $\alpha>0$,
\begin{eqnarray*}
&&P\left(
\varepsilon_n^{1/2} \frac{n}{\kappa_n(x)}\left|\hat c_{n,2}(x)-c \right|
\left|\hat{f}_n(x)-f(x)\right|\geq\eta\right)\\
&\leq& P\left(\left|\hat c_{n,2}(x)-c \right|\geq c/\alpha \right)\\
&+&
P\left(
\varepsilon_n^{1/2} \frac{nc}{\kappa_n(x)}
\left|\hat{f}_n(x)-f(x)\right|\geq\eta\alpha\right).
\end{eqnarray*}
Thus, from (C.2) and the part a) of the proof,
\begin{eqnarray*}
&&\mathop{\lim\sup}_{n\to\infty}\varepsilon_n\log P\left(
\varepsilon_n^{1/2} \frac{n}{\kappa_n(x)}\left|\hat c_{n,2}(x)-c \right|
\left|\hat{f}_n(x)-f(x)\right|\geq\eta\right) \\
&\leq & \mathop{\lim\sup}_{n\to\infty}\varepsilon_n\log P\left(
\varepsilon_n^{1/2} \frac{nc}{\kappa_n(x)}
\left|\hat{f}_n(x)-f(x)\right|\geq\eta\alpha\right) \\
&\leq& -\frac{\alpha^2\eta^2}{2\sigma^2(x)}.
\end{eqnarray*}
Letting $\alpha\to\infty$ gives the intended result~(\ref{butun}).  
\ \rule{0.5em}{0.5em}

\noindent\textbf{Proof of Corollary \ref{coromain} and Corollary \ref{coromain2}:} Mimicking the
part b) of the proof of Theorem~\ref{TCLvarphi} and Theorem~\ref{devvarphi},
it is sufficient to prove that
\begin{equation}
\label{butdeux}
\mathop{\lim\sup}_{n\to\infty}\varepsilon_n\log P\left(
\left|\varphi(x,\hat{f}_n(x))-\varphi(x,f(x))\right|\geq\eta\right)=-\infty.  
\end{equation}
Since, from ($\Phi$), $\varphi$ is continuous at point $f(x)$, there exists $\delta>0$
such that 
$$
|y-f(x)|< \delta \Rightarrow |\varphi(x,f(x))-\varphi(x,\hat f_n(x))|< \eta.  
$$
Hence, for all large $\alpha>0$,
\begin{eqnarray*}
&&\mathop{\lim\sup}_{n\to\infty}\varepsilon_n\log P\left(
\left|\varphi(x,\hat{f}_n(x))-\varphi(x,f(x))\right|\geq\eta\right)\\
&\leq&  \mathop{\lim\sup}_{n\to\infty}\varepsilon_n\log P\left(
\varepsilon_n^{1/2} \frac{nc}{\kappa_n(x)}
\left|\hat{f}_n(x)-f(x)\right|\geq\delta\varepsilon_n^{1/2} \frac{nc}{\kappa_n(x)}
\right)\\
&\leq&  \mathop{\lim\sup}_{n\to\infty}\varepsilon_n\log P\left(
\varepsilon_n^{1/2} \frac{nc}{\kappa_n(x)}
\left|\hat{f}_n(x)-f(x)\right|\geq\alpha\right)\\
&\leq& -\frac{\alpha^2}{2\sigma^2(x)}.
\end{eqnarray*}
Letting $\alpha\to\infty$ gives the intended result~(\ref{butdeux}).  
\ \rule{0.5em}{0.5em}

\subsection{Proofs of section~\ref{secpolar}}

Proofs of Theorem~\ref{TCLpol} and Theorem~\ref{MDpol} both rely on the 
following lemma which permits to apply Theorem~\ref{TCLvarphi} and Theorem~\ref{devvarphi}. 
\begin{lemma}
\label{lemvarphi}
In situation (P) (resp. (E)), $\left(  X_{n,i},Y_{n,i}\right)  _{i\geq1}$ is a Poisson 
(resp. an empirical) process 
with mean measure $
n\;c \;\gamma_d \;y^{d-1}\;  \mathbf{1}_{S}\left(  x,y\right)\; h_d(x)\;dx\; dy$.
\end{lemma}
\begin{proof} 
Note that the Jacobian of the inverse polar transformation $P_d^{-1}$ is
$$
J(x,y)= y^{d-1}\prod_{j=1}^{d-1}(\sin x_j)^{d-1-j}
=\gamma_d \;y^{d-1}\; h_d(x).
$$
Hence, in situation (P), the result follows from the Mapping Theorem
(see Kingman~(1993), p.~18).
In situation (E), the result is obtained by a change of variable
(see Cohn~(1980), Theorem~6.1.6). 
\end{proof}

\begin{lemma}
\label{lemme4}
With (\ref{eqmesurepol}), (\ref{choixnoyau1}) and (\ref{choixcn}), estimator~(\ref{defest}) can be
rewritten as
\[
\hat{f}_{n}^{\mathrm{pol}}(x)=\left(
\sum_{r=1}^{k_n}\left(\int_{I_{n,r}}K_{n}(x,t)\nu(dt)
+ \frac{\int_E K_n(x,t)\nu(dt)}{k_n N_{n,r}}\right)
\left( Y_{n,r}^{\ast}\right)  ^{d}\right)  ^{1/d}.
\]
\end{lemma}
\noindent\textbf{Proof.}
From~(\ref{defest}) and (\ref{eqmesurepol}), 
$$
\hat{f}_{n}^{\mathrm{pol}}(x)=\left(
\overset{k_{n}}{\underset{r=1}{\sum}}\kappa_{n,r}(x)\left(
\nu_{n,r}  \left( Y_{n,r}^{\ast}\right)  ^{d}+ \frac{d}{\gamma_d}\frac{1}{ n\hat c_n}\right)\right)  ^{1/d}.
$$
Thus, taking account of (\ref{choixnoyau1}) and (\ref{choixcn}),
\begin{eqnarray*}
\hat{f}_{n}^{\mathrm{pol}}(x)&=&\left(
\sum_{r=1}^{k_n}\int_{I_{n,r}}K_{n}(x,t)\nu(dt) \left(
\left( Y_{n,r}^{\ast}\right)  ^{d}+ \frac{1}{k_n}\sum_{s=1}^{k_n}
 \frac{\left( Y_{n,s}^{\ast}\right) ^{d}}{N_{n,s}}\right)\right)  ^{1/d}\\
 &=&\left(
\sum_{r=1}^{k_n}\left(\int_{I_{n,r}}K_{n}(x,t)\nu(dt)
+ \frac{\int_E K_n(x,t)\nu(dt)}{k_n N_{n,r}}\right)
\left( Y_{n,r}^{\ast}\right)  ^{d}\right)  ^{1/d}.
\end{eqnarray*}
\ \rule{0.5em}{0.5em}

\noindent{\bf Proof of Theorem~\ref{TCLpol} and Theorem~\ref{MDpol}:}
First, Lemma~4.5 in Menneteau~(2007) shows that, under (H.1) and (H.2),
conditions (C.1) and (C.2) hold for $\hat c_n$ defined in~(\ref{choixcn}).
Second, in the proofs of Theorem~3.1 and Theorem~3.2 in Menneteau~(2007), it
is shown that conditions (H.1), (H.2), (K.1)-(K.6) imply conditions
(H.1)-(H.6) of Theorem~\ref{TCLvarphi} and Theorem~\ref{devvarphi} and 
that $\kappa_n(x)=k_n^{1/2} \|K_n(x,.)\|_2(1+o(1))$.
Thus, (K.7) implies (H.7).
\ \rule{0.5em}{0.5em}


\noindent{\bf Proof of Corollary~\ref{coropol}:}
From Menneteau~(2007), Corollary~3.8,
(i)--(v) imply (H.1), (H.2) and (K.1)-(K.6).
Moreover it is clear that (i), (ii) give (K.7) since,
by Tolstov~(1976),
$\|K_n^D(x,.)\|_2=(\ell_n+1)^{1/2}$. 
\ \rule{0.5em}{0.5em}

\newpage

\noindent{\large\bf References}
\begin{description}

\item{ Baillo, A. and Cuevas, A.} (2001)
{ On the estimation of a star-shaped set.}
{\em Adv. Appl. Proba.}, {\bf 33}(4), 717--726.

\item{ Chevalier J.} (1976)
{ Estimation du support et du contour d'une loi de probabilit\'e}
{\em Ann. Inst. H. Poincar\'e}, sect. B, {\bf 12}, 4, 339--364.

\item{Cohn, D.} (1980)
\textit{Measure theory.}
Birkh\"auser.

\item{Dembo, A. and Zeitouni, O. } (1993)
\textit{Large Deviations Techniques and Applications.}
Jones and Bartlett, Boston and London.

\item{ Deprins, D., Simar, L. and Tulkens, H.} (1984)
{ Measuring Labor Efficiency in Post Offices.} in
{\em The Performance of Public Enterprises: Concepts and Measurements} by
M. Marchand, P. Pestieau and H. Tulkens, North Holland ed, Amsterdam.

\item{ Devroye, L.P. and Wise, G.L.} (1980)
{ Detection of abnormal behavior via non parametric estimation of the support.}
{\em SIAM J. Applied Math.}, {\bf 38}, 448--480.

\item{ Geffroy, J.} (1964)
{ Sur un probl\`eme d'estimation g\'eom\'etrique.}
{\em Publications de l'Institut de Statistique de l'Universit\'e de Paris},
 {\bf XIII}, 191--200.

\item{ Gensbittel, M.H., (1979)} 
{\em Contribution \`a l'\'etude statistique de r\'epartitions ponctuelles
al\'eatoires. }
PhD Thesis, Universit\'e Pierre et Marie Curie, Paris.

\item { Girard, S. and Menneteau, L.} (2005)
Central limit theorems for smoothed extreme value estimates
of Poisson point processes boundaries.
\textit{Journal of Statistical Planning and Inference}, {\bf 135}, 433--460.

\item{ Hall, P., Park, B. U. and Stern, S. E.} (1998)
{ On polynomial estimators of frontiers and boundaries.}
{\em J. Multiv. Analysis}, {\bf 66}, 71--98.

\item{ Hardy A. and Rasson J. P.} (1982)
{ Une nouvelle approche des probl\`emes de classification automatique.}
{\em Statistique et Analyse des donn\'ees}, {\bf 7}, 41--56.

\item {Kallenberg, W.} (1983a)
Intermediate efficiency, theory and examples. 
\textit{Annals of Statistics}, {\bf 11}, 170--182. 

\item {Kallenberg, W.} (1983b)
On moderate deviation theory in estimation. 
\textit{Annals of Statistics}, {\bf 11}, 498--504. 

\item{Kingman, J.} (1993)
\textit{Poisson processes.}
Oxford Studies in Probability, {\bf 3}, Oxford: Clarendon Press.

\item{  Korostelev, A. P. and Tsybakov, A. B.} (1993)
{ Minimax theory of image reconstruction.}
{\em Lecture Notes in Statistics}, {\bf 82}, Springer Verlag, New-York.

\item{Mardia, K., Kent, J. and Bibby, J.} (1979)
\textit{Multivariate analysis.}
Academic Press London.

\item { Menneteau, L.} (2007)
Multidimensional limit theorems for smoothed extreme value estimates
of point process boundaries.
\textit{ESAIM: Probability and Statistics}, to appear.

\item{ Renyi, A. and Sulanke, R.} (1963)
{  Uber die konvexe H\"ulle von n zuf\"alligen gew\"ahlten Punkten.}
{\em Z. Wahrscheinlichkeitstheorie verw. Geb.} {\bf 2}, 75--84.

\item{ Stone, C.} (1982)
Optimal global rates of convergence for nonparametric regression, {\it{Annals of Statistics}}, {\bf{10}}, 1040--1053.

\item{Tolstov, G.P. } (1976)
\textit{Fourier Series.}
2nd ed., Dover Publications, New York.

\end{description}

\vskip .65cm
\noindent
$^{(1)}$ Team MISTIS, INRIA Rh\^one-Alpes, ZIRST, 655 avenue
de l'Europe, \\
38330 Montbonnot, France, {\tt Stephane.Girard@inrialpes.fr}
\vskip 2pt
\noindent
$^{(2)}$ EPS/I3M, Universit\'e Montpellier 2,
place Eug\`ene Bataillon, \\
34095 Montpellier cedex 5, France, {\tt mennet@math.univ-montp2.fr}

\begin{figure}[p]
\begin{minipage}{0.45\textwidth}
\centerline{\epsfig{figure=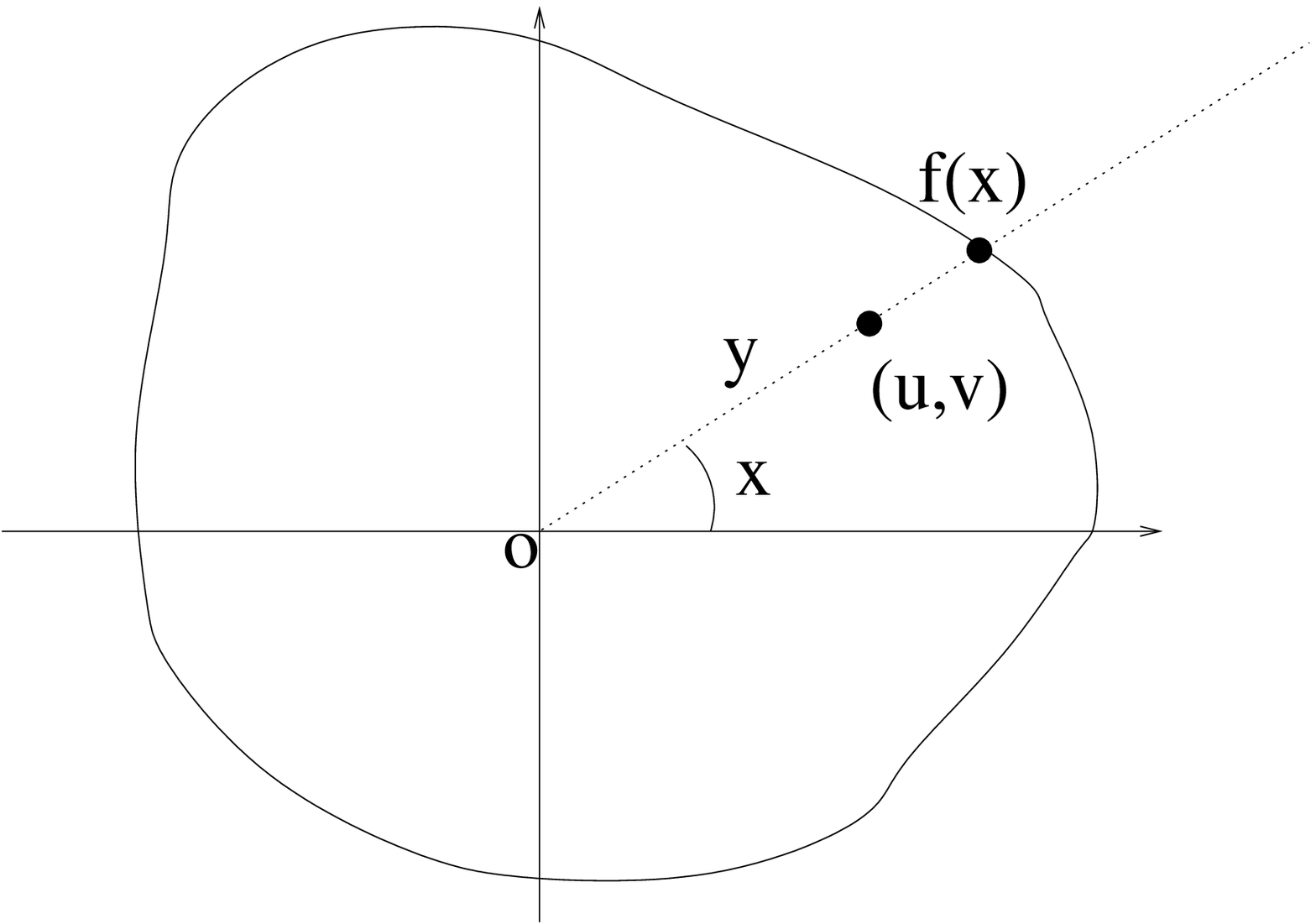,width=0.90\textwidth}}
\end{minipage}
\begin{minipage}{0.45\textwidth}
\centerline{\epsfig{figure=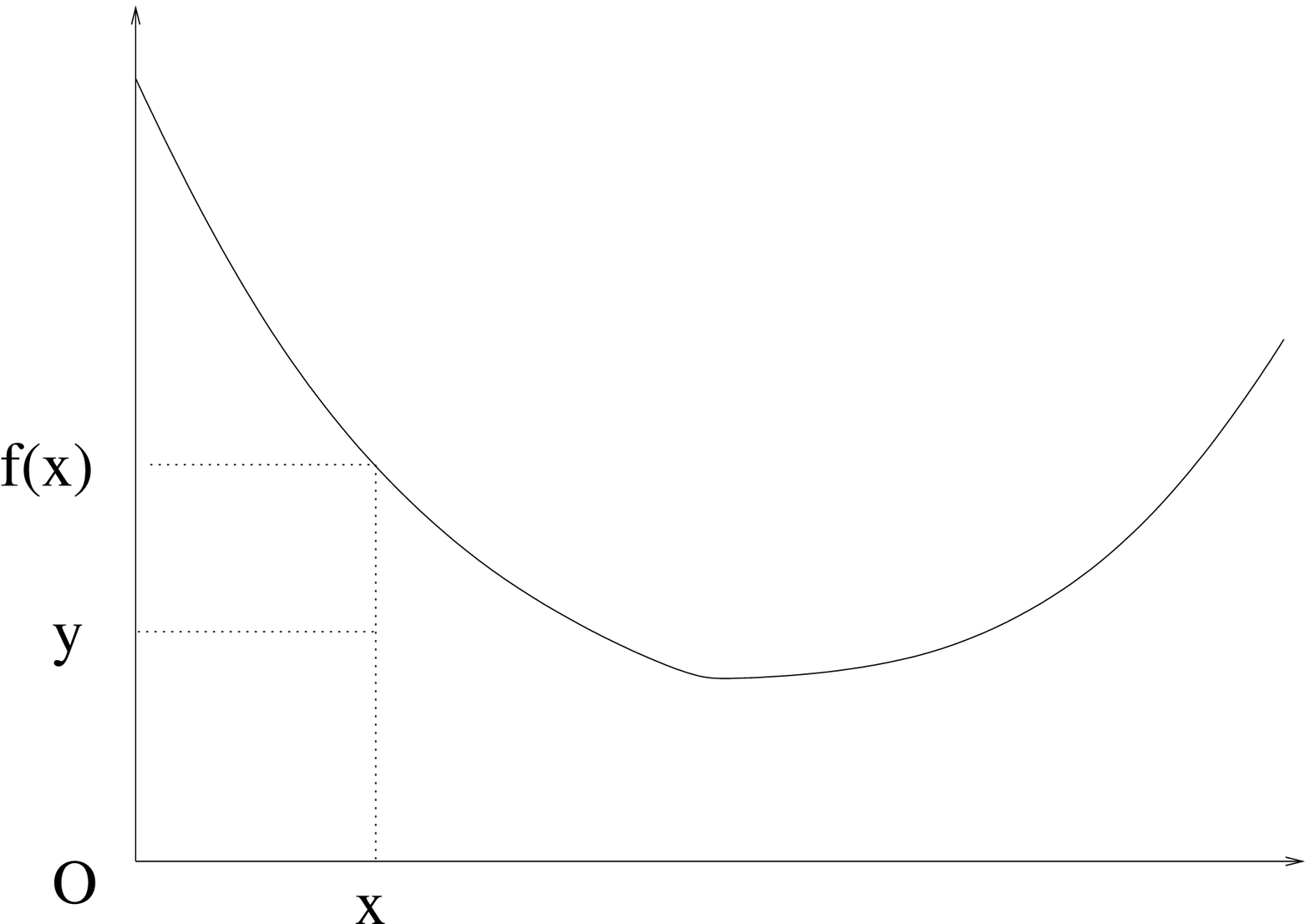,width=0.90\textwidth}}
\end{minipage}
\caption{Two different parametrizations of the point process support. 
Left: $S^{\mathrm{pol}}$ is described with polar coordinates, 
right: $S$ is described in Cartesian coordinates.   }
\label{figdessin}
\end{figure}

\clearpage

\begin{figure}[p]
\centerline{\epsfig{figure=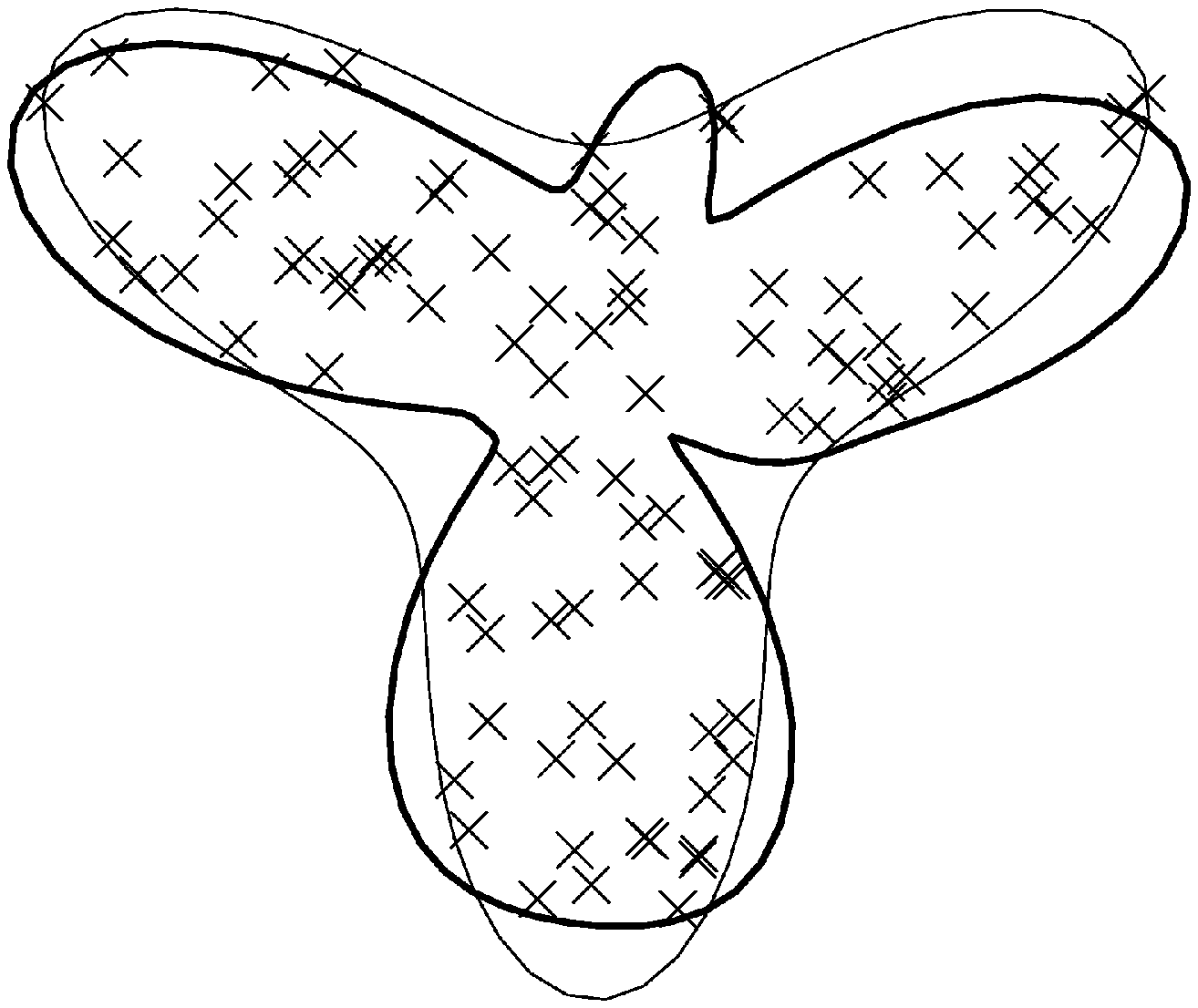,width=0.80\textwidth,angle=270}}
\caption{Worst situation.  Thin line:~$f$, bold line:~$\hat f^{\mathrm{pol}}_n$.  }
\label{figsim2}
\end{figure}
\begin{figure}[p]
\centerline{\epsfig{figure=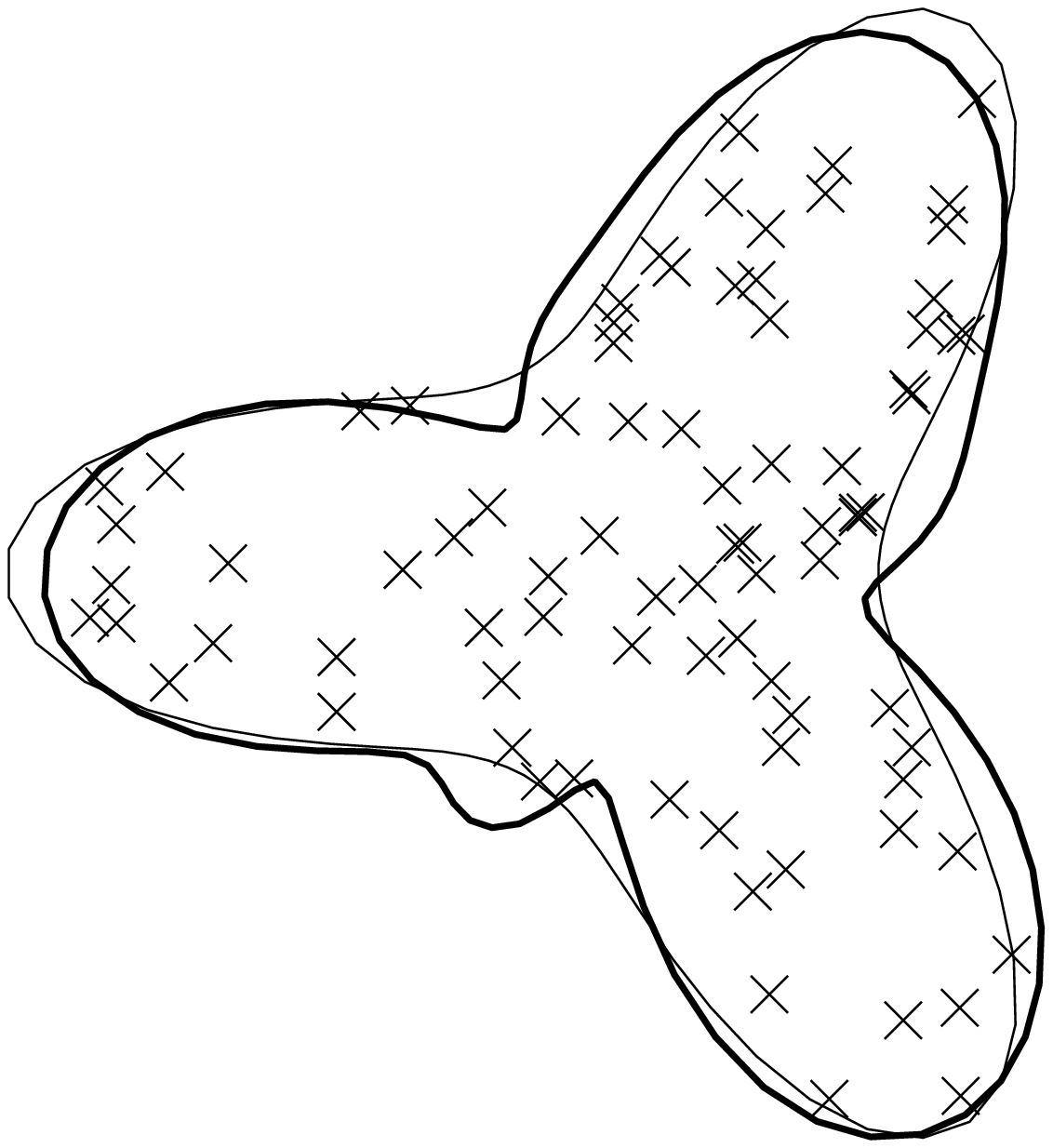,width=0.80\textwidth,angle=270}}
\caption{Best situation.  Thin line:~$f$, bold line:~$\hat f^{\mathrm{pol}}_n$.  }
\label{figsim}
\end{figure}

\end{document}